\documentclass[conference]{IEEEtran}
\usepackage{graphicx,amssymb,amsmath}
\usepackage{multicol}
\usepackage[noadjust]{cite}
\usepackage{setspace}
\usepackage{stfloats,subfigure}
\usepackage{midfloat}
\usepackage[normal]{threeparttable}
\usepackage{amsthm}
\usepackage{cuted}
\usepackage{cases,subeqnarray}
\usepackage{bm,multirow,bigstrut}
\usepackage{textcomp}
\usepackage{latexsym,bm}
\usepackage{booktabs,changebar}
\usepackage{xcolor}
\usepackage{mathtools}
\usepackage{dsfont}
\usepackage{extarrows}
\theoremstyle{plain}
\newtheorem{lemm}{Lemma}

\newtheorem{coro}{Corollary}

\IEEEoverridecommandlockouts 
\begin{document}

\title{Multipair Massive MIMO Two-Way Full-Duplex Relay Systems with Hardware Impairments}
%
\author{Ying~Liu,
        Xipeng~Xue,
        Jiayi~Zhang,
        Xu~Li,
        Linglong Dai,
        and~Shi Jin

\thanks{This work was supported in part by the National Natural Science Foundation of China (Grant No. 61601020) and the Fundamental Research Funds for the Central Universities (Grant Nos. 2016RC013, 2017JBM319, and 2016JBZ003). (Corresponding author: jiayizhang@bjtu.edu.cn)}%
\thanks{Y. Liu, X. Xue, J. Zhang and X. Li are with the School of Electronic and Information Engineering, Beijing Jiaotong University, Beijing 100044, P. R. China.}
\thanks{L. Dai is with the Department of Electronic Engineering, Tsinghua University, Beijing 100084, P. R. China.}
\thanks{S. Jin is with the National Mobile Communications Research Laboratory, Southeast University, Nanjing 210096, P. R. China.}
}


\maketitle
\begin{abstract}
Hardware impairments, such as phase noise, quantization errors, non-linearities, and noise amplification, have baneful effects on wireless communications. In this paper, we investigate the effect of hardware impairments on multipair massive multiple-input multiple-output (MIMO) two-way full-duplex relay systems with amplify-and-forward scheme. More specifically, novel closed-form approximate expressions for the spectral efficiency are derived to obtain some important insights into the practical design of the considered system. When the number of relay antennas $N$ increases without bound, we propose a hardware scaling law, which reveals that the level of hardware impairments that can be tolerated is roughly proportional to $\sqrt{N}$. This new result inspires us to design low-cost and practical multipair massive MIMO two-way full-duplex relay systems. Moreover, the optimal number of relay antennas is derived to maximize the energy efficiency. Finally, Motor-Carlo simulation results are provided to validate our analytical results.
\end{abstract}


\section{Introduction}
The two-way full-duplex (FD) relay system can ideally achieve almost twice of the spectral efficiency (SE) achieved by the traditional two-way half-duplex (HD) scheme, since the relay can transmit and receive signals simultaneously. However, the practical implementation of the two-way FD relay is challenging due to the severe self-interference (SI) caused by FD \cite{wong2017key, ngo2014multipair,Xie2014channel}.

Recently, massive multiple-input multiple-output (MIMO) has been proposed as an efficient approach to suppress the SI of two-way FD relay systems in the spatial domain \cite{zhang2016JSAC}. Different from most of existing works, which consider systems deploying high-cost ideal hardware components, in this paper we consider a multipair massive MIMO two-way FD relay system with low-cost non-ideal hardware that suffers from hardware impairments. In practical systems, the cost and power consumption increase with the number of radio frequency (RF) chains. In order to achieve higher energy efficiency (EE) and/or lower hardware cost, each RF chain can use some cheap hardware components \cite{zhang2017performance,zhang2016spectral}. However, low-cost hardware is particularly prone to the impairments in transceivers, such as quantization errors of low-resolution analog to digital converters (ADCs), I/Q-imbalance, phase noise, and non-linearities \cite{zhang2016achievable,Bjornson2014TIT,zhang2016constrain}.

Although the influence of hardware impairments can be mitigated by some compensation algorithms, residual impairments still exist due to time-varying and random hardware characteristics. The effect of hardware impairments on the massive MIMO two-way FD relay system has only recently been studied in \cite{xia2015hardware}, which focused on the decode-and-forward (DF) scheme at the relay. However, the signal processing complexity of the DF scheme is much higher than that of the amplify-and-forward (AF) scheme for the implementation of massive MIMO relay systems \cite{chen2015large}. Therefore, the AF scheme is more attractive for practical system design. To the best of authors' knowledge, the performance of AF based multipair massive MIMO two-way FD relay systems with hardware impairments has not been investigated in the literature, partially due to the difficulty of manipulating products of SI and hardware impairments vectors.

Motivated by the aforementioned consideration, a natural question is that, whether the low-cost non-iedeal hardware can be deployed at the AF based multipair massive MIMO two-way FD relay system without sacrificing the expected performance gains? In this paper, we try to answer this question with the following contributions:
\begin{itemize}
\item An analytical SE approximation of multipair massive MIMO two-way FD relay systems with hardware impairments is derived in closed-form. The effect of both the number of relay antennas and the level of transceiver hardware impairments on the SE has been investigated.

\item A hardware scaling law has been presented to show that one can tolerate larger level of hardware impairments as the number of antennas increases. This is an analytic proof that the considered system can be deployed with low-cost hardware components.

\item The sum SEs of FD and HD systems have been compared with different levels of hardware impairments. It is interesting to find that the FD system with hardware impairments can achieve the same SE of the HD system with larger loop interferences. Finally, we derive the optimal number of relay antennas to maximize the EE.
\end{itemize}



\section{System Model}\label{se:system}

%
We consider a massive MIMO two-way FD relay system where $K$ pairs of devices on two sides communicate with each other through a single relay ${T_R}$. The devices are denoted as $ {{T_{{A_i}}}} $ and ${{T_{{B_i}}}} $, for $i = 1,\ldots,K $, respectively. The devices could be sensors that exchange a small amount of information or small cell base stations which need high throughput links. The relay is equipped with $2N$ antennas, where $N$ antennas are used for transmission, and the other $N$ antennas are used for reception. Each device is equipped with one receive and one transmit antenna. In addition, the relay and devices are assumed to work in FD mode, i.e., they can transmit and receive signals at the same time. We further assume that there is no direct communication link between each pair of devices due to heavy shadowing and/or large path loss. The device is interfered by other devices on the same side.

\subsection{Channel Model}
Block fading is considered in this paper. This means that it is an ergodic process with a static channel realization in a coherence block and the realizations in blocks are independent. Then, we define ${\bf{G_u}} \buildrel \Delta \over = \left[ {{{\bf{g}}_{u1}},\ldots,{{\bf{g}}_{uK}}} \right]$ and ${\bf{H_u}} \buildrel \Delta \over = \left[ {{{\bf{h}}_{u1}},\ldots,{{\bf{h}}_{uK}}} \right]$, where ${{\bf{g}}_{ui}} \in {\mathbb{C}^{N \times 1}}$ and ${{\bf{h}}_{ui}}\in {\mathbb{C}^{N \times 1}}$ $\left(i = 1,\ldots,K \right)$, which denote the uplink channels between ${T_{{A_i}}}$ and ${T_R}$, ${T_{{B_i}}}$ and ${T_R}$, respectively. In addition, the downlink channels between ${T_{{A_i}}}$ and ${T_R}$, ${T_{{B_i}}}$ and ${T_R}$ are given by ${\bf{G_d}} \buildrel \Delta \over = \left[ {{{\bf{g}}_{d1}},\ldots,{{\bf{g}}_{dK}}} \right]$ and ${\bf{H_d}} \buildrel \Delta \over = \left[ {{{\bf{h}}_{d1}},\ldots,{{\bf{h}}_{dK}}} \right]$, where ${{\bf{g}}_{di}} \in {\mathbb{C}^{N \times 1}}$ and ${{\bf{h}}_{di}}\in {\mathbb{C}^{N \times 1}} $ $\left(i = 1,\ldots,K \right)$, respectively. Furthermore, $\bf{G_u}$, $\bf{H_u}$, $\bf{H_d}$ and $\bf{H_d}$ are assumed to follow the independent identically distributed (i.i.d.) Rayleigh fading, i.e., the elements of ${{\bf{g}}_{uK}}$, ${{\bf{h}}_{uK}}$, ${{\bf{g}}_{dK}}$, and ${{\bf{h}}_{dK}}$ are i.i.d. ${\cal{CN}}({\bf{0}},\sigma _{{g_{ui}}}^2)$, ${\cal{CN}}({\bf{0}},\sigma _{{h_{ui}}}^2)$, ${\cal{CN}}({\bf{0}},\sigma _{{g_{di}}}^2)$, and ${\cal{CN}}({\bf{0}},\sigma _{{h_{di}}}^2)$ random variables \cite{zhang2016spectral}. Furthermore, $\bf{G_u}$, $\bf{H_u}$, $\bf{G_d}$, and $\bf{H_d}$ can be expressed as ${\bf{G_u}} = {{\bf{S}}_{gu}}{\bf{D}}_{gu}^{1/2}$, ${\bf{H_u}} = {{\bf{S}}_{hu}}{\bf{D}}_{hu}^{1/2}$, ${\bf{G_d}} = {{\bf{S}}_{gd}}{\bf{D}}_{gd}^{1/2}$, and ${\bf{H_d}} = {{\bf{S}}_{hd}}{\bf{D}}_{hd}^{1/2}$, respectively, where ${{\bf{S}}_{gu}}$, ${{\bf{S}}_{hu}}$, ${{\bf{S}}_{gd}}$, and ${{\bf{S}}_{hd}}$ stand for the small-scale fading and their elements are all i.i.d. ${\cal{CN}}({\bf{0}},{\bf{1}})$ random variables. On the other hand, ${{\bf{D}}_{gu}}$, ${{\bf{D}}_{hu}}$, ${{\bf{D}}_{gd}}$, and ${{\bf{D}}_{hd}}$ are diagonal matrices representing the large-scale fading, and the $k$th diagonal elements are denoted as $\sigma _{{g_{uk}}}^2$, $\sigma _{{h_{uk}}}^2$, $\sigma _{{g_{dk}}}^2$, and $\sigma _{{h_{dk}}}^2$, respectively.

Furthermore, let ${{\bf{G}}_{RR}} \in {\mathbb{C}^{N \times N}}$ denote the SI matrix between the transmit and receive arrays of the relay due to the FD mode. Each row of ${{\bf{G}}_{RR}}$ such as ${\bf{G}}_{RRi}$ denotes the channel between $i$th receive antenna and all transmit antennas of the relay. ${\Omega _{k,i}}$ is the inter-device interference channel coefficient from $i$th device to $k$th device. Note that ${\Omega _{k,k}}$ denotes the SI at the $k$th device. The elements of ${{\bf{G}}_{RR}}$ and ${\Omega _{k,i}}$ are random variables following the i.i.d. complex Gaussian distribution, e.g., ${\cal{CN}}({\bf{0}},\sigma _{{LIr}}^2)$ and ${\cal{CN}}({\bf{0}},\sigma _{{k,i}}^2)$, respectively \cite{ngo2014multipair}.

\subsection{Hardware Impairments}

As shown in \cite{bjornson2014massive}, the residual hardware impairments at the transmitter and receiver can be modeled as additive distortion noises that are proportional to the signal power. Thus, the additive distortion term ${{\boldsymbol{\eta} _r}} $ describes the residual impairments of receiver at the relay and is proportional to the instantaneous power of received signals at the relay antenna as ${\boldsymbol{\eta} _r} \sim {\cal{CN}}\left({\bf{0}},\kappa _r^2 {\tt{diag}} \left({{W}}_{11}, \dots,{{W}}_{NN} \right)\right)$, where $W_{ii}$ is the $i$th diagonal element of the covariance matrix ${\bf{W}} = \sum_{j=1}^{K} {P_U}({\bf{h}}_{uj} {\bf{h}}_{uj}^H +{\bf{g}}_{uj} {\bf{g}}_{uj}^H)+\frac{{{P_R}}}{N}\sum_{j=1}^{N} {{\bf{G}}_{RRj}{\bf{G}}_{RRj}^H}$ with $P_U$ being the power constraint of the device and ${P_R}$ being the transmit power of relay \cite{bjornson2014massive}. Furthermore, the proportionality coefficient $\kappa _r$ describes the level of hardware impairments and is related to the received error vector magnitude (EVM) \cite{Bjornson2014TIT}. Note that the EVM is a common quality indicator of the signal distortion magnitude, and it can be defined as the ratio of the signal distortion to the signal magnitude. For example, the EVM at relay can be defined as \cite[Eq. (5)]{Bjornson2014TIT}
\begin{align}\label{definition of EVM}
EV{M_r} = \sqrt {\frac{{E\left\{ {{{\left\| {{{\bf{\eta }}_r}} \right\|}^2}\left| \Im  \right.} \right\}}}{{E\left\{ {{{\left\| {\bf{x}} \right\|}^2}\left| \Im  \right.} \right\}}}}  = \sqrt {\frac{{tr\left( {\kappa _r^2{\bf{W}}} \right)}}{{tr\left( {\bf{W}} \right)}}}  = {\kappa _r},
\end{align}
where $\Im$ denotes the set of channel realizations (i.e., ${{\bf{g}}_u},{{\bf{h}}_u} \in \Im $). Furthermore, 3GPP LTE suggests that the EVM should be smaller than 0.175 \cite{Bjornson2014TIT}.

\subsection{Signal Transmission}
At the time instant $n$, all devices ${T_{{A_i}}}$ and ${T_{{B_i}}}$ $\left(i = 1,\ldots,K \right)$ transmit their signals ${{x_{{A_i}}}\left( n \right)}$ and ${{x_{{B_i}}}\left( n \right)}$ to the relay ${T_R}$, respectively, and ${T_R}$ broadcasts its processed previously received signal ${{{\bf{y}}_t}\left( n \right)}$ to all devices.

First, we assume that ${{x_{{A_i}}}\left( n \right)}$ and ${{x_{{B_i}}}\left( n \right)}$ are Gaussian distributed signals. Due to the FD mode, ${T_R}$ also receives the signal, i.e., ${{{\bf{y}}_t}}\left( n \right)$ which is broadcasted to all devices. Thus, at the time instant $n$, the received signal at ${T_R}$ is given by
\begin{align}\label{receive signals at relay}
{{\bf{y}}_r}\left( n \right) = {\bf{Ax}}\left( n \right) + {{\bf{G}}_{RR}}{{\bf{y}}_t}\left( n \right) + {{\boldsymbol{\eta }}_r} + {{\bf{n}}_R}\left( n \right),
\end{align}
where ${\bf{A}} \buildrel \Delta \over = \left[ {{{\bf{G}}_u},{{\bf{H}}_u}} \right]$, ${{\bf{x}}\left( n \right)} \buildrel \Delta \over = {\left[ {{\bf{x}}_A^T\left( n \right),{\bf{x}}_B^T\left( n \right)} \right]^T}$ with
${{\bf{x}}_A\left( n \right)} \buildrel \Delta \over = \left[ {{x_{{A_1}}\left( n \right)},\ldots,{x_{{A_K}}\left( n \right)}} \right]$ and ${{\bf{x}}_B\left( n \right)} \buildrel \Delta \over = \left[ {{x_{{B_1}}\left( n \right)},\ldots,{x_{{B_K}}}\left( n \right)} \right]$, and ${{{\bf{n}}}_R\left( n \right)} \sim {\cal{CN}}(0,\sigma _R^2{{\bf{I}}_N})$ denotes an additive white Gaussian noise (AWGN) vector at ${T_R}$.

Then, we analyze the received signal at devices. At the time instant $n$ $\left( {n > 1} \right)$, the relay using the simple AF protocol amplifies the previously received signal ${{\bf{y}}_r}\left( {n - 1} \right)$ and broadcasts it to the devices. Therefore, the transmit signal vector at the relay is given by
\begin{align}\label{transimition signals at relay}
{{{\bf{y}}}_t^{'}}\left( n \right) = \rho {\bf{F}}{{\bf{y}}_r}\left( {n - 1} \right),
\end{align}
where ${\bf{F}} \in {\mathbb{C}^{N \times N}} $ is the precoding matrix and ${\rho}$ is the amplification factor. Then ${T_R}$ broadcasts ${{{\bf{y}}}_t^{'}}\left( n \right)$ to all devices. However, due to the hardware impairments of RF chains at the transmitter, ${T_R}$ actually broadcasts ${{\bf{y}}_t}\left( n \right)$ to all devices as
\begin{align}\label{real transimition signals at relay}
{{\bf{y}}_t}\left( n \right) = {{{\bf{y}}}_t^{'}}\left( n \right) + {{\boldsymbol{\eta }}_t} = \rho {\bf{F}}{{\bf{y}}_r}\left( {n - 1} \right) + {{\boldsymbol{\eta }}_t},
\end{align}
where ${\boldsymbol{\eta} _t} \sim {\cal{CN}}\left({\bf{0}},\kappa _t^2\frac{{P_R}}{N} {{\bf{I}}_N}\right)$ with the proportionality parameters $\kappa _t$ characterizing the level of hardware impairment at the transmitter. Here we assume that each antenna has the same power and ${T_R}$ can obtain perfect channel state information (CSI) according to uplink pilots from the devices, and the devices can then obtain CSI through channel reciprocity \cite{ngo2013energy}. Due to the power constraint of the relay $P_R$, ${\rho}$ is normalized by the instantaneous received signal power
\begin{align}
\rho  = \sqrt {\frac{{{P_R}}}{{{P_U}{{\left\| {{\bf{FA}}} \right\|}^2} + \frac{{{P_R}}}{N}{{\left\| {{\bf{F}}{{\bf{G}}_{RR}}} \right\|}^2} + {{\left\| {{\bf{F}}{{\bf{\eta }}_r}} \right\|}^2} + \sigma _R^2{{\left\| {\bf{F}} \right\|}^2}}}} .\notag
\end{align}

At the relay, we adopt the low-complexity maximum-ratio (MR) scheme suitable for low-cost massive MIMO deployment \cite{ngo2013energy}. Therefore, the precoding matrix ${\bf{F}}$ can be written as ${\bf{F}} = {{\bf{B}}^*}{{\bf{A}}^H}$,
where ${\bf{B}} \buildrel \Delta \over = \left[ {{{\bf{H}}_d},{{\bf{G}}_d}} \right]$, ${{\bf{G}}_d} \buildrel \Delta \over = \left[ {{{\bf{g}}_{d1}}, \ldots ,{{\bf{g}}_{dK}}} \right]$ and ${{\bf{H}}_d} \buildrel \Delta \over = \left[ {{{\bf{h}}_{d1}}, \ldots ,{{\bf{h}}_{dK}}} \right]$.
To our best knowledge, it is very challenging to analyze the residual loop interference power if substituting \eqref{real transimition signals at relay} into \eqref{receive signals at relay} iteratively. However, the residual loop interference can be modeled as additional Gaussian noise. This is due to the fact that the loop interference can be significantly degraded and the residual loop interference is too weak by applying loop interference mitigation schemes \cite{zhang2016JSAC}.

Following similar steps in \cite{zhang2016JSAC}, ${{\bf{y}}_r}\left( n \right)$ in \eqref{receive signals at relay} can be approximated by a Gaussian noise source ${{{\bf{\tilde y}}}_r}\left( n \right)$ with $  {{\mathbb{E}}\left\{ {{{{\bf{\tilde y}}}_r}\left( n \right){\bf{\tilde y}}_r^H\left( n \right)} \right\} = \frac{{{P_R}}}{N}{{\bf{I}}_N}}  $. Furthermore, ${T_{{A_i}}}$ and ${T_{{B_i}}}$ receive the combined signal as
\begin{align}
{Z_{{A_i}}}\left( n \right) = {\bf{g}}_{di}^T{{\bf{y}}_t}\left( n \right) + \sum\limits_{i,k \in {U_A}} {{\Omega _{i,k}}{x_{{A_k}}}\left( n \right)}  + {n_{{A_i}}}\left( n \right) ,\label{expression of zai}\\
{Z_{{B_i}}}\left( n \right) = {\bf{h}}_{di}^T{{\bf{y}}_t}\left( n \right) + \sum\limits_{i,k \in {U_B}} {{\Omega _{i,k}}{x_{{B_k}}}\left( n \right)}  + {n_{{B_i}}}\left( n \right),
\end{align}
where the noise ${n_{{A_i}}\left( n \right)}$ and ${n_{{B_i}}\left( n \right)}$ are AWGN with ${n_{{A_i}}\left( n \right)} \sim {\cal{CN}}(0,\sigma _{{A_i}}^2)$ and ${n_{{B_i}}\left( n \right)} \sim {\cal{CN}}(0,\sigma _{{B_i}}^2)$, respectively. In the following, we only discuss the analytical result for ${T_{{A_i}}}$. The corresponding result of ${T_{{B_i}}}$ can be obtained by replacing ${T_{{A_i}}}$ with ${T_{{B_i}}}$. Note that the relay can only receive signal and the transmission part keeps silent at the first time slot $\left( n=1 \right)$, during which the received signals at the relay and devices are respectively given by
\begin{align}\label{first time slot}
&{{\bf{y}}_r}\left( 1 \right) = {\bf{Ax}}\left( 1 \right) + {{\boldsymbol{\eta }}_r} + {{\bf{n}}_R}\left( 1 \right)\\
&{Z_{{A_i}}}\left( 1 \right) = {\bf{g}}_{di}^T{{\bf{y}}_t}\left( 1 \right) + {n_{{A_i}}}\left( 1 \right),i = 1, \ldots ,K.
\end{align}
For simplicity, the time label $n$ is omitted in the following \cite{zhang2016spectral}. Substituting \eqref{receive signals at relay} and \eqref{real transimition signals at relay} into \eqref{expression of zai}, the combined received signal ${Z_{{A_i}}} $ can be expressed as
\begin{align}\label{expression of expanded zai}
&{Z_{{A_i}}} =  \underbrace {\rho {\bf{g}}_{di}^T{\bf{F}}{{\bf{h}}_{ui}}{x_{{B_i}}}}_{\text {desired signal}}\!+\!\underbrace {\rho \sum\limits_{j = 1,j \ne i}^K {\left( {{\bf{g}}_{di}^T{\bf{F}}{{\bf{g}}_{uj}}{x_{{A_j}}}\! + \! {\bf{g}}_{di}^T{\bf{F}}{{\bf{h}}_{uj}}{x_{{B_j}}}} \right)} }_{\text  {inter-pair interference}} \notag \\
&+\underbrace {\rho {\bf{g}}_{di}^T{\bf{F}}{{\bf{g}}_{ui}}{x_{{A_i}}}}_{\text {self-interference}} + \underbrace {\rho {\bf{g}}_{di}^T{\bf{F}}{{\bf{G}}_{RR}}{{{\bf{\tilde y}}}_r}}_{\text {loop interference}}  + \underbrace {\sum\limits_{i,k \in {U_A}} {{\Omega _{i,k}}{x_{{A_k}}}} }_{\text { inter-device interference by FD mode}} \notag \\
&+ \underbrace {\rho {\bf{g}}_{di}^T{\bf{F}}{{\boldsymbol{\eta }}_r} + {\bf{g}}_{di}^T{{\boldsymbol{\eta }}_t}}_{\text {hardware impairments}} + \underbrace {\rho {\bf{g}}_{di}^T{\bf{F}}{{\bf{n}}_R} + {n_{{A_i}}}}_{\text {compound noise}},
\end{align}
where we use the set notation of ${U_A} = \left\{ {1,3, \ldots ,2K - 1} \right\}$ or ${U_B} = \left\{ {2,4, \ldots ,2K} \right\}$ to represent the devices on bothsides of relay. Note that one set of devices (${U_A}$) can not exchange information with the other set (${U_B}$) directly. From \eqref{expression of expanded zai}, we can find that ${Z_{{A_i}}}$ is composed of seven terms: the signal that ${T_{{A_i}}}$ desires to receive, the inter-pair interference due to other devices' signal, the SI from the device, the loop interference from the relay, the inter-device interference caused by other devices due to FD mode, the distortion noise induced by hardware impairments at the relay, and the compound noise.

With the power constraint of the relay and perfect CSI, the FD relay can take advantage of massive antennas and simple SI cancellation (SIC) schemes to eliminate the SI \cite{zhang2016JSAC}. Furthermore, the interference and noise power can be obtained by taking expectation with respect to interference and noise within one coherence block of channel fading. As a result, the SE of ${T_{{A_i}}}$ is given by
\begin{align}\label{ergodic achievable rate}
{R_{{A_i}}} = {\mathbb{E}}\left\{ {{{\log }_2}\left( {1 + {\text{SINR}}_{{A_i}}} \right)} \right\},\;\;\; {\text{for}} \ i = 1,\ldots,K,
\end{align}
where ${\text{SINR}}_{{A_i}}$ denotes the signal-to-interference plus noise ratio (SINR) of $A_i$ and can be expressed as
\begin{align}
{\text{SINR}}_{{A_i}} \!=\! \frac{{{P_U}{{\left| {{\bf{g}}_{di}^T{\bf{F}}{{\bf{h}}_{ui}}} \right|}^2}}}{{A' \!+\! B' \!+\! C' \!+\! D' \!+\! E' \!+\! {{\left| {{\bf{g}}_{di}^T{\bf{F}}{{\boldsymbol{\eta }}_r}} \right|}^2} \!+\! \frac{1}{{{\rho ^2}}}{{\left| {{\bf{g}}_{di}^T{{\boldsymbol{\eta }}_t}} \right|}^2}}},\notag
\end{align}
where $A' \buildrel \Delta \over = {P_U}\sum\limits_{j = 1,j \ne i}^K {\left( {{{\left| {{\bf{g}}_i^T{\bf{F}}{{\bf{g}}_j}} \right|}^2} + {{\left| {{\bf{g}}_i^T{\bf{F}}{{\bf{h}}_j}} \right|}^2}} \right)} $, $B' \buildrel \Delta \over = \sigma _R^2{\left\| {{\bf{g}}_i^T{\bf{F}}} \right\|^2}$, $C' \buildrel \Delta \over = \frac{{\sigma _{Ai}^2}}{{{\rho ^2}}}$, $D' \buildrel \Delta \over = {\left| {{\bf{g}}_i^T{\bf{F}}{{\bf{G}}_{RR}}{{{\bf{\tilde y}}}_r}} \right|^2}$, and $E' \buildrel \Delta \over = \frac{{{P_U}}}{{{\rho ^2}}}\sum\limits_{i,k \in {U_A}} {\sigma _{i,k}^2} $, respectively.

\section{Performance Analysis}\label{analysis}
To the best of authors' knowledge, the exact derivation of \eqref{ergodic achievable rate} is really difficult \cite{jin2015ergodic}. Herein we consider the asymptotic scenario when $N \to \infty$, which is the large system limit. Utilizing the convexity of $\log_2{(1+1/x)}$ and Jensen's inequality, the lower bound of ${R_{{A_i}}}$ in \eqref{ergodic achievable rate} can be written as
\begin{align}\label{ergodic achievable rate low bound}
{R_{{A_i}}} \ge {\tilde R_{{A_i}}}   = {\log _2}\left( {1 + \frac{1}{{\mathbb{E}}\left\{ {{{\left[ {{\text{SINR}}_{{A_i}}} \right]}^{ - 1}}} \right\}}} \right).
\end{align}
Based on \eqref{ergodic achievable rate low bound} and considering devices at both sides, we can obtain the sum SE of the multipair massive MIMO two-way FD relay system as
\begin{align}\label{sum SE}
R_\text{sum} = \sum\limits_{i = 1}^K {\left( {\tilde R_{{A_i}}} + {\tilde R_{{B_i}}} \right)} .
\end{align}

Note that, in the remainder of the paper, we only show the analytical results for ${R_{{A_i}}}$ since the formula of ${R_{{B_i}}}$ is symmetric with that of ${R_{{A_i}}}$.
In the following, we present the SE of $A_i$ in \textbf{Lemma 1}.
\begin{lemm}\label{lemm1}
With hardware impairments and MR processing at the relay, ${\tilde  R_{{A_i}}}$ can be approximated as
\begin{align}\label{An approximate ergodic achievable rate}
{\tilde R_{{A_i}}}\!-\! {\log _2}\left( {1 \!+\! \frac{N}{{{A_i} \!+\! {B_i} \!+\! {C_i} \!+\! {D_i} \!+\! {E_i} \!+\! {F_i} \!+\! {G_i}}}} \right)\!\xrightarrow[{N \to \infty }]{}\! 0,
\end{align}
where ${A_i} \buildrel \Delta \over = \sum\limits_{j = 1,j \ne i}^K {\left( {\frac{{\sigma _{{h_{uj}}}^2}}{{\sigma _{{h_{ui}}}^2}} + \frac{{\sigma _{{h_{uj}}}^4\sigma _{{g_{dj}}}^2}}{{\sigma _{{h_{ui}}}^4\sigma _{{g_{di}}}^2}} + \frac{{\sigma _{{g_{uj}}}^2}}{{\sigma _{{h_{ui}}}^2}} + \frac{{\sigma _{{g_{uj}}}^4\sigma _{{h_{dj}}}^2}}{{\sigma _{{h_{ui}}}^4\sigma _{{g_{di}}}^2}}} \right)} $, ${B_i} \buildrel \Delta \over = \frac{{\sigma _R^2}}{{{P_U}\sigma _{{h_{ui}}}^2}}$, ${C_i} \buildrel \Delta \over = \frac{{\sigma _{{A_i}}^2J}}{{{P_R}{P_U}\sigma _{{g_{di}}}^4\sigma _{{h_{ui}}}^4}}$, ${D_i} \buildrel \Delta \over = \frac{{\kappa _r^2\left( {{P_U}\sum\limits_{j = 1}^K {\left( {\sigma _{{h_{uj}}}^2 + \sigma _{{g_{uj}}}^2} \right)}  + {P_R}\sigma _{LIr}^2} \right)}}{{{P_U}\sigma _{{h_{ui}}}^2}}$, ${E_i} \buildrel \Delta \over = \frac{{\kappa _t^2J}}{{{P_U}\sigma _{{g_{di}}}^2\sigma _{{h_{ui}}}^4}}$, ${F_i} \buildrel \Delta \over = \frac{{{P_R}\sigma _{LIr}^2}}{{{P_U}\sigma _{{h_{iu}}}^2}}$, ${G_i} \buildrel \Delta \over = \frac{{J\sum\limits_{i,k \in {U_A}} {\sigma _{i,k}^2} }}{{{P_R}\sigma _{{g_{di}}}^4\sigma _{{h_{ui}}}^4}}$, and \\
$J \buildrel \Delta \over = {P_U}\sum\limits_{i = 1}^K {\left( {\sigma _{{g_{ui}}}^4\sigma _{{h_{di}}}^2 \!+ \!\sigma _{{g_{di}}}^2\sigma _{{h_{ui}}}^4} \right)} \! +\! \frac{{\kappa _r^2}}{N}\sum\limits_{i = 1}^K {\left( {\sigma _{{g_{ui}}}^2\sigma _{{h_{di}}}^2 \!+\! \sigma _{{g_{di}}}^2\sigma _{{h_{ui}}}^2} \right)}\left( {{P_U}\sum\limits_{j = 1}^K {\left( {\sigma _{{g_{uj}}}^2\! + \!\sigma _{{h_{uj}}}^2} \right)} \! +\! {P_R}\sigma _{LIr}^2} \right)$.
\end{lemm}

\begin{proof}
Please refer to Appendix.
\end{proof}
From \textbf{Lemma 1}, it is clear to see that the SE ${R_{{A_i}}}$ increases with the number of antennas $N$. Further insights can be gained by investigating the terms ${A_i}$, ${B_i}$, ${C_i}$, ${D_i}$, ${E_i}$, ${F_i}$, and ${G_i}$ in \eqref{An approximate ergodic achievable rate}, respectively. First, we focus on the inter-device interference term ${A_i}$ caused by the broadcasting signal from the relay. The SE ${R_{{A_i}}}$ increases when we enlarge the values of ${\sigma _{{g_{di}}}^2}$ and ${\sigma _{{h_{ui}}}^2}$, which indicates that reducing the channel fading of the $i$th device pair. However, ${R_{{A_i}}}$ will decrease if we enlarge ${\sigma _{{g_{uj}}}^2}$ and ${\sigma _{{h_{uj}}}^2}$, for $j \ne i$, which means that reducing the channel fading of other device pairs except the $i$th device pair. This finding is consistent with the result in \cite{jin2015ergodic}.

Furthermore, \textbf{Lemma \ref{lemm1}} reveals that ${B_i}$ consists of the transmit power of ${T_{{B_i}}}$ and the channel fading ${\sigma _{{h_{ui}}}^2}$ from ${T_{{B_i}}}$ to ${T_R}$. Therefore, we can increase the transmit power of ${T_{{B_i}}}$ and/or decrease ${\sigma _{{h_{ui}}}^2}$ to increase ${R_{{A_i}}}$. Then, from ${C_i}$, we can find that ${R_{{A_i}}}$ increases when the transmit power for ${T_R}$ and the transmit power of devices increase, but decreases when $\rho$ becomes large. Moreover, it is clear to see from \eqref{An approximate ergodic achievable rate} that the detrimental effect of hardware impairments in ${D_i}$, ${E_i}$ and ${G_i}$ on the SE ${R_{{A_i}}}$. Finally, the loop interference due to the FD mode in ${F_i}$ and ${G_i}$ can also reduce the SE.

In order to show how fast the hardware impairments can increase with $N$ while maintaining the constant rate, we establish an important hardware scaling law in the following corollary.
\begin{coro}\label{coro:1}
Suppose the hardware impairment parameters are replaced by $\kappa _{\rm{r}}^{\rm{2}} = \kappa _{{\rm{0r}}}^{\rm{2}}{N^z}$ and $\kappa _{\rm{t}}^{\rm{2}} = \kappa _{{\rm{0t}}}^{\rm{2}}{N^z}$ for an initial value $\kappa_{\rm{0r}} \geq 0$, $\kappa_{\rm{0t}} \geq 0$ and a given scaling exponent $0 < z \le 1$, the SE ${\tilde  R_{{A_i}}}$, under MR processing and $N \to \infty$, converges to a non-zero limit
\begin{align}\label{Corollary 1 RA}
\left\{ \begin{array}{l}
\tilde R_{{\rm{A}}i}^{{\rm{}}} - {\log _2}\left( {1 + \frac{{\sigma _{{{\rm{g}}_{{\rm{d}}i}}}^2\sigma _{{{\rm{h}}_{{\rm{u}}i}}}^4{N^{1 - z}}}}{{\kappa _{{\rm{0r}}}^2\sigma _{{{\rm{g}}_{{\rm{d}}i}}}^2\sigma _{{{\rm{h}}_{{\rm{u}}i}}}^{\rm{2}}\xi  + {\rm{2}}K\kappa _{{\rm{0t}}}^{\rm{2}}{{\tilde \mu }^{{\rm{}}}}}}} \right)\xrightarrow[{N \to \infty }]{}0,0 < z < 1\\
\tilde R_{{\rm{A}}i}^{{\rm{}}} - {\log _2}\left( {1 + \frac{{\sigma _{{{\rm{g}}_{{\rm{d}}i}}}^2\sigma _{{{\rm{h}}_{{\rm{u}}i}}}^4}}{{\kappa _{{\rm{0r}}}^{\rm{2}}\left( {{\rm{2}}K\kappa _{{\rm{0t}}}^{\rm{2}}{{\bar \mu }^{{\rm{}}}} + \sigma _{{{\rm{g}}_{{\rm{d}}i}}}^2\sigma _{{{\rm{h}}_{{\rm{u}}i}}}^{\rm{2}}} \right)\xi  + {\rm{2}}K\kappa _{{\rm{0t}}}^{\rm{2}}{{\tilde \mu }^{{\rm{}}}}}}} \right)\\
\xrightarrow[{N \to \infty }]{}0,z = 1
\end{array} \right.
\end{align}
where $\xi  \buildrel \Delta \over = {\rm{2}}K{\mu ^{{\rm{}}}} + {P_{\rm{R}}}\sigma _{{\rm{LIr}}}^2/{P_{\rm{U}}}$, ${\mu ^{{\rm{}}}} \buildrel \Delta \over = \sum\nolimits_{j = 1}^K {\left( {\sigma _{{{\rm{h}}_{{\rm{u}}j}}}^2 + \sigma _{{{\rm{g}}_{{\rm{u}}j}}}^2} \right)} /\left( {2K} \right)$,
${\tilde \mu ^{{\rm{}}}} \buildrel \Delta \over = \sum\nolimits_{i = 1}^K {\left( {\sigma _{{{\rm{g}}_{{\rm{u}}i}}}^4\sigma _{{{\rm{h}}_{{\rm{d}}i}}}^2 + \sigma _{{{\rm{g}}_{{\rm{d}}i}}}^2\sigma _{{{\rm{h}}_{{\rm{u}}i}}}^4} \right)} /\left( {2K} \right)$ and ${\bar \mu ^{{\rm{}}}} \buildrel \Delta \over = \sum\nolimits_{i = 1}^K {\left( {\sigma _{{{\rm{g}}_{{\rm{u}}i}}}^2\sigma _{{{\rm{h}}_{{\rm{d}}i}}}^2 + \sigma _{{{\rm{g}}_{{\rm{d}}i}}}^2\sigma _{{{\rm{h}}_{{\rm{u}}i}}}^2} \right)} /\left( {2K} \right)$.
\end{coro}

\begin{proof}
Substituting $\kappa _{\rm{r}}^{\rm{2}} = \kappa _{{\rm{0r}}}^{\rm{2}}{N^z}$ and $\kappa _{\rm{t}}^{\rm{2}} = \kappa _{{\rm{0t}}}^{\rm{2}}{N^z}$ into \eqref{An approximate ergodic achievable rate}, with  $N \to \infty$ and $0 < z \le 1$, ${A_i}$, ${B_i}$, ${C_i}$, ${F_i}$, and ${G_i}$ tend to zero. Moreover, ${D_i}$ behaves as $\mathcal{O}(N^z)$, while ${E_i}$ behaves as $\mathcal{O}(N^{z}+N^{2z-1})$. To make the numerator and denominator have the identical scaling, we can finish the proof by fulfilling $1-\max(z,2z-1) \geq 0$ as $ 0 < z\leq 1$.
\end{proof}

\textbf{Corollary \ref{coro:1}} reveals that large level of hardware impairments can be compensated by increasing number of antennas at the relay in multipair massive MIMO two-way relaying systems. Furthermore, the EVM at relay is defined as ${\text{EVM}}  = \kappa$ \cite{bjornson2014massive}. Considering the condition of $\kappa^2=\kappa_0^2N^z$ in \textbf{Corollary 1} and $z=1$, it is easy to have $\text{EVM}^2=\kappa_0^2N$, which means the EVM can be increased proportionally to $N^{1/2}$. Thus, for the negligible SE loss, we can replace 8 high-quality antennas with $\text{EVM}=0.05$ with 128 low-quality antennas with $\text{EVM}=0.2$. This encouraging result enable reducing the power consumption and cost of the multipair massive MIMO two-way FD relay system.

In the following, we evaluate the EE of the multipair massive MIMO two-way FD relay system when the number of relay antennas becomes large. The EE is defined as the ratio of the sum SE to the total power consumption of the system \cite{ngo2013energy}. Considering the classical architecture where each antenna is connected to one RF chain. The total power consumption of the system can be modeled as \cite{zhang2016JSAC}
\begin{align}\label{P_total}
{P_\text{total}} = \left( {N + 2K} \right)\left( {{P_t} + {P_r}} \right) + {P_0} + \left( {2K{P_U} + {P_R}} \right)/\varphi,
\end{align}
where $P_t$ and $P_r$ are the power of RF chains at the transmitter and receiver, respectively. Moreover, $P_0$ denotes the power of the static circuits, the term $2K P_U + P_R$ is the total power of the power amplifiers at devices and the relay, and $\varphi$ denotes the efficiency of the power amplifier in each RF chain. Thus, the EE of the considered system is given by $EE = {R_\text{sum}}/{{{P_\text{total}}}}$.

\section{Numerical Results}\label{numerical}
In this section, the derived results of the multipair massive MIMO two-way FD relay system with hardware impairments and AF schemes are validated through Monte-Carlo simulations by averaging over $10^4$ independent channel samples. Similar to previous works \cite{jin2015ergodic,zhang2016JSAC}, we set $P_u=10$ W, $P_R=40$ W and normalize $\sigma _{{R}}^2=\sigma _{{A_i}}^2=\sigma _{{B_i}}^2=1$ for $i = 1,\ldots,K$. Furthermore, without loss of generality, we simply set the same values for both loop and inter-device interferences as $\sigma^2=\sigma _{\text{LIr}}^2 =\sigma _{k,i}^2=1$ $\left(i \in {U_A} \cup {U_B},k = 1, \ldots ,2K \right)$ and $\kappa_0=\kappa _{{\rm{0r}}}=\kappa _{{\rm{0t}}}$ , respectively \cite{ngo2014multipair,zhang2016spectral}.

\begin{figure}[t]
\centering
\includegraphics[scale=0.65]{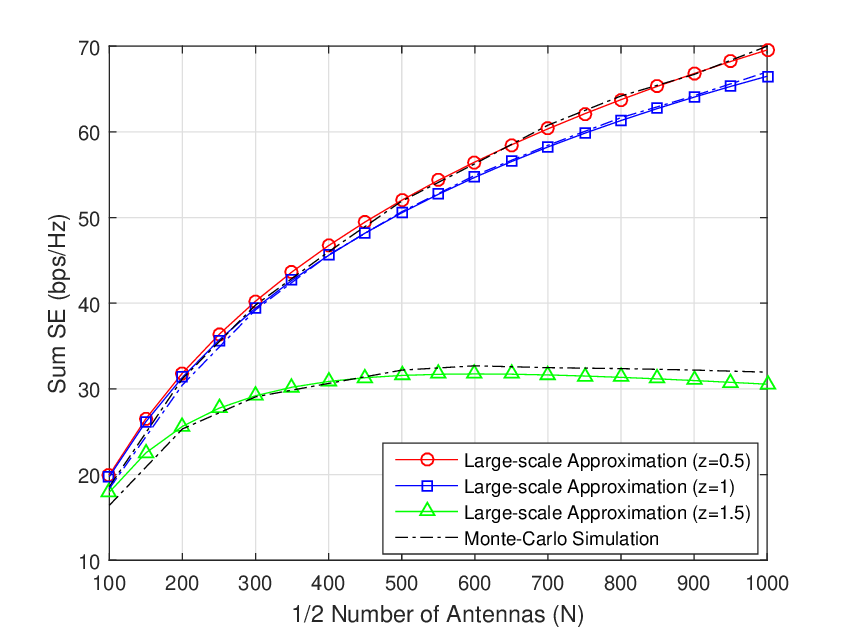}
\caption{Hardware scaling law of multipair massive MIMO two-way FD relay systems against different number of antennas $N$ at the relay ($K=10$, $\kappa_0=0.0156$).}
\label{Fig.1}
\end{figure}

The simulated and analytical asymptotic sum SE \eqref{An approximate ergodic achievable rate} are plotted as a function of the half number of antennas $N$ at the relay in Fig. \ref{Fig.1}. The simulation results validate the tightness of the derived large-scale approximations. Moreover, Fig. \ref{Fig.1} validates the hardware scaling law established by \textbf{Corollary \ref{coro:1}}. The SE grows with low levels of hardware impairments $(z=0.5, 1)$. However, the SE curve asymptotically bend toward zero when the scaling law is not satisfied ($z=1.5$).

Note that the analytical curves plotted in Fig. \ref{Fig.1} are not always below the simulated curves. This is due to the reason that we utilize the large number law to derive the SE. When $N$ is relatively small (e.g. $N<350$), the low order term of $N$ cannot be omitted. Thus, the analytical result is a little larger than the corresponding simulation result. However, even if $N$ is smaller, the curves of large-scale approximation and Monte-Carlo simulations are close \cite{jin2015ergodic}. Compared with the DF scheme in \cite{xia2015hardware}, the signal processing can achieve a smaller sum SE. However, the complexity of AF based system at the relay is much lower than the DF based system.

\begin{figure}[t]
\centering
\includegraphics[scale=0.65]{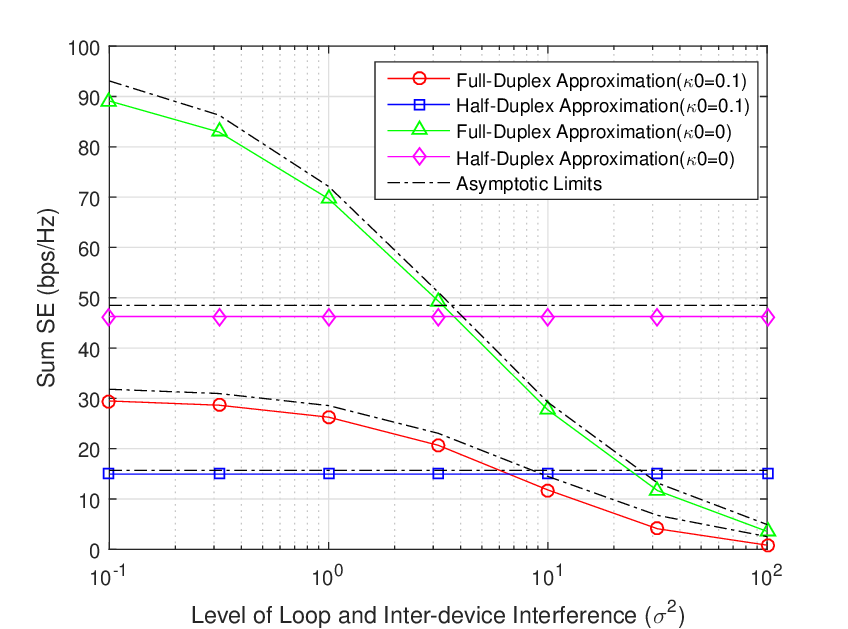}
\caption{Sum SE of multipair massive MIMO two-way FD relay systems with hardware impairments against different levels of loop and inter-device interference ($N=1000$, $z=1$, $K=10$).}
\label{Fig.4}
\end{figure}

Fig. \ref{Fig.4} shows the large-scale approximation \eqref{An approximate ergodic achievable rate} and the asymptotic SE limit \eqref{Corollary 1 RA} against the levels of the loop and inter-device interferences $\sigma^2$. The SE of such system in HD model is also plotted as a baseline for comparison. Since the HD system utilizes two phases to transmit and receive signal, the inherent loop and inter-device interference do not exist. Therefore, the SE of HD systems is constant in Fig. \ref{Fig.4}. The first observation from Fig. \ref{Fig.4} is that when $N=1000$, the asymptotic SE limit of FD systems outperforms the one of HD systems for small and moderate levels of interference, e.g., $\sigma^2<10^{1.18}$ for ideal hardware ($\kappa_0=0$) and $\sigma^2<10^{0.9}$ for non-ideal hardware ($\kappa_0=0.1$). This can be explained that only the half time required in the FD mode compared with the HD mode. Interestingly, the SE of the multipair massive MIMO two-way FD relay system with hardware impairments is larger than the one with ideal hardware for small and moderate levels of loop and inter-device interference. However, large value of loop and inter-device interference decreases the sum SE of FD systems. Moreover, the gap of SE curves between FD and HD systems increases with the level hardware impairments $\kappa_0$.

\begin{figure}[t]
\centering
\includegraphics[scale=0.65]{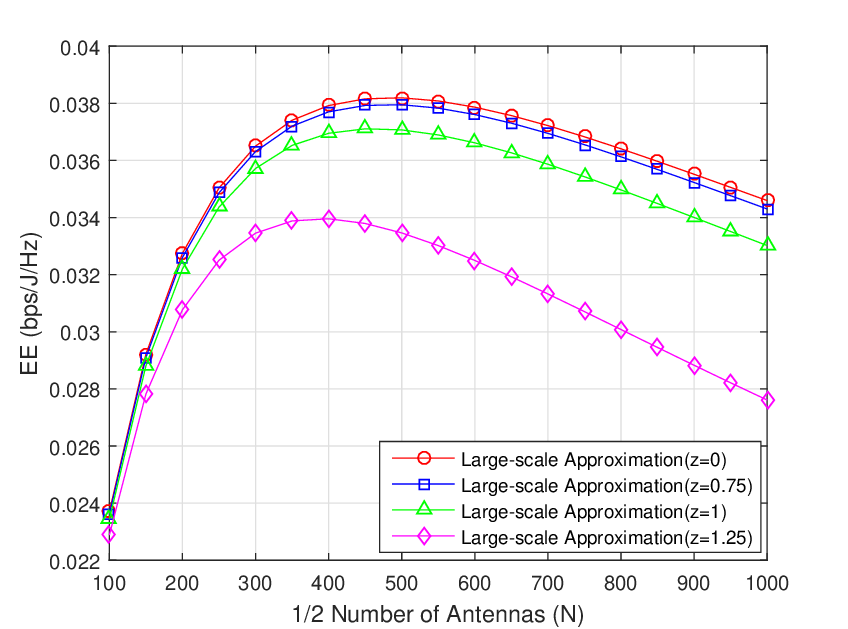}
\caption{EE of multipair massive MIMO two-way FD relay systems with hardware impairments against different number of antennas $N$ at the relay ($K=10$, $\kappa_0=0.0156$).}
\label{Fig.5}
\end{figure}

The large-scale approximation of EE as a function of the number of antennas $N$ at the relay is plotted in Fig. \ref{Fig.5}. Similar to \cite{Li2013Throughput}, we set $P_t=1 $ W, $P_r=0.3 $ W, $P_0=2$ W and $\varphi=0.35$. It is clear to see that the EE decreases with $z$ due to the distortion noise caused by hardware impairments. Moreover, there exists an optimal number of antennas $N_\text{opt}$ to reach the corresponding maximum EE. When $N \leq N_\text{opt}$, the EE can be improved by increasing $N$. However, when $N > N_\text{opt}$, increasing $N$ will reduce the EE since the addition power consumption of RF chains and static circuits dominate the performance.


\section{Conclusions}\label{se:Conclusion}
In this paper, we investigate the SE and EE of AF-based multipair massive MIMO two-way FD relay systems with hardware impairments. The effect of $N$ and $\kappa$ on the SE has been investigated by deriving a closed-form large-scale approximate expression. In addition, the optimal number of relay antennas has been derived to maximize the EE. We also find that the SE of the massive MIMO two-way FD system with hardware impairments outperforms that of the HD system when the level of loop and inter-device interference is small and moderate. Finally, an useful hardware scaling law has been established to prove that low-cost hardware can be deployed at the relay due to the huge degrees-of-freedom brought by massive antennas.

\vspace{-0.5cm}
\appendix\label{sec:proof}
From \eqref{ergodic achievable rate low bound}, we can rewrite ${\mathbb{E}}\left\{ {{{\left[ {{\text{SINR}}_{{A_i}}} \right]}^{ - 1}}} \right\}$ as
\begin{align}\label{expression of EAi 1}
&{\mathbb{E}}\left\{ {\frac{1}{{SIN{R_{{A_i}}}}}} \right\} = \frac{{{P_R}}}{{N{P_U}}}{\mathbb{E}}\left\{ {\frac{{{{\left\| {{\bf{g}}_{di}^T{\bf{F}}{{\bf{G}}_{RR}}} \right\|}^2}}}{{{{\left| {{\bf{g}}_{di}^T{\bf{F}}{{\bf{h}}_{ui}}} \right|}^2}}}} \right\} \notag \\
&+ \frac{{\sigma _R^2}}{{{P_U}}}{\mathbb{E}}\left\{ {\frac{{{{\left\| {{\bf{g}}_{di}^T{\bf{F}}} \right\|}^2}}}{{{{\left| {{\bf{g}}_{di}^T{\bf{F}}{{\bf{h}}_{ui}}} \right|}^2}}}} \right\} + \frac{1}{{{\rho ^2}{P_U}}}{\mathbb{E}}\left\{ {\frac{{{{\left| {{\bf{g}}_{di}^T{{\bf{\eta }}_t}} \right|}^2}}}{{{{\left| {{\bf{g}}_{di}^T{\bf{F}}{{\bf{h}}_{ui}}} \right|}^2}}}} \right\} \notag \\
&+ \frac{{\sum\limits_{i,k \in {U_A}} {\sigma _{i,k}^2} }}{{{\rho ^2}}}{\mathbb{E}}\left\{ {\frac{1}{{{{\left| {{\bf{g}}_{di}^T{\bf{F}}{{\bf{h}}_{ui}}} \right|}^2}}}} \right\} + \frac{{\sigma _{{A_i}}^2}}{{{\rho ^2}{P_U}}}{\mathbb{E}}\left\{ {\frac{1}{{{{\left| {{\bf{g}}_{di}^T{\bf{F}}{{\bf{h}}_{ui}}} \right|}^2}}}} \right\}\notag \\
&+\sum\limits_{j = 1,j \ne i}^K {\left( {{\mathbb{E}}\left\{ {\frac{{{{\left| {{\bf{g}}_{di}^T{\bf{F}}{{\bf{g}}_{uj}}} \right|}^2}}}{{{{\left| {{\bf{g}}_{di}^T{\bf{F}}{{\bf{h}}_{ui}}} \right|}^2}}}} \right\} + {\mathbb{E}}\left\{ {\frac{{{{\left| {{\bf{g}}_{di}^T{\bf{F}}{{\bf{h}}_{uj}}} \right|}^2}}}{{{{\left| {{\bf{g}}_{di}^T{\bf{F}}{{\bf{h}}_{ui}}} \right|}^2}}}} \right\}} \right)}\notag \\
&+ \frac{1}{{{P_U}}}{\mathbb{E}}\left\{ {\frac{{{{\left| {{\bf{g}}_{di}^T{\bf{F}}{{\bf{\eta }}_r}} \right|}^2}}}{{{{\left| {{\bf{g}}_{di}^T{\bf{F}}{{\bf{h}}_{ui}}} \right|}^2}}}} \right\}.
\end{align}
According to the law of large numbers, we have
\begin{align}
&\frac{1}{N}{\bf{g}}_{di}^T{\bf{F}}{{\bf{h}}_{ui}} \!-\! \frac{1}{N}{\left\| {{\bf{g}}_{di}^*} \right\|^2}{\left\| {{{\bf{h}}_{ui}}} \right\|^2} \!\xrightarrow[{N \to \infty }]{}  0,\notag\\
&\frac{1}{N}{\bf{g}}_{di}^T{\bf{F}}{{\bf{h}}_{uj}} \!-\! \frac{1}{N}\left( {{{\left\| {{\bf{g}}_{di}^*} \right\|}^2}{\bf{h}}_{ui}^H{{\bf{h}}_{uj}} \!+\! {{\left\| {{{\bf{h}}_{uj}}} \right\|}^2}{\bf{g}}_{di}^T{\bf{g}}_{dj}^*} \right) \!\xrightarrow[{N \to \infty }]{}  0,\notag\\
& \frac{1}{N}{\bf{g}}_{di}^T{\bf{F}}{{\bf{g}}_{uj}} \!-\! \frac{1}{N}\left( {{{\left\| {{\bf{g}}_{di}^*} \right\|}^2}{\bf{h}}_{ui}^H{{\bf{g}}_{uj}} \!+\! {{\left\| {{{\bf{g}}_{uj}}} \right\|}^2}{\bf{g}}_{di}^T{\bf{h}}_{dj}^*} \right)\!\xrightarrow[{N \to \infty }]{}  0,\notag\\
&\frac{1}{N}{\left\| {{\bf{g}}_{di}^T{\bf{F}}} \right\|^2} \!-\! \frac{1}{N}{\left\| {{\bf{g}}_{ui}^*} \right\|^4}{\left\| {{{\bf{h}}_{di}}} \right\|^2}\! \xrightarrow[{N \to \infty }]{}  0.\notag
\end{align}
When $N \to \infty$, we further have
\begin{align}
&{\mathbb{E}}\left\{ {\frac{{{\bf{h}}_{ui}^H{{\bf{h}}_{uj}}}}{{{{\left\| {{{\bf{h}}_{ui}}} \right\|}^2}}} \!+\! \frac{{{{\left\| {{{\bf{h}}_{uj}}} \right\|}^2}{\bf{g}}_{di}^T{\bf{g}}_{dj}^*}}{{{{\left\| {{{\bf{h}}_{ui}}} \right\|}^2}{{\left\| {{\bf{g}}_{di}^*} \right\|}^2}}}} \right\} \!-\! \frac{1}{N}\left( {\frac{{\sigma _{{h_{uj}}}^2}}{{\sigma _{{h_{ui}}}^2}} \!+\! \frac{{\sigma _{{h_{uj}}}^4\sigma _{{g_{dj}}}^2}}{{\sigma _{{h_{ui}}}^4\sigma _{{g_{di}}}^2}}} \right)\! \xrightarrow[{N \to \infty }]{}  0,\notag\\
&{\mathbb{E}}\left\{ {\frac{{{\bf{h}}_{ui}^H{{\bf{g}}_{uj}}}}{{{{\left\| {{{\bf{h}}_{ui}}} \right\|}^2}}} \!+\!  \frac{{{{\left\| {{{\bf{g}}_{uj}}} \right\|}^2}{\bf{g}}_{di}^T{\bf{h}}_{dj}^*}}{{{{\left\| {{\bf{g}}_{di}^*} \right\|}^2}{{\left\| {{{\bf{h}}_{ui}}} \right\|}^2}}}} \right\} \!-\!  \frac{1}{N}\left( {\frac{{\sigma _{{g_{uj}}}^2}}{{\sigma _{{h_{ui}}}^2}} \!+\!  \frac{{\sigma _{{g_{uj}}}^4\sigma _{{h_{dj}}}^2}}{{\sigma _{{h_{ui}}}^4\sigma _{{g_{di}}}^2}}} \right)\xrightarrow[{N \to \infty }]{}  0,\notag\\
&{\mathbb{E}}\left\{ {\frac{1}{{{{\left\| {{{\bf{h}}_{ui}}} \right\|}^2}}}} \right\} - \frac{1}{{N\sigma _{{h_{ui}}}^2}} \xrightarrow[{N \to \infty }]{}  0,\notag\\
&{\mathbb{E}}\left\{ {\frac{{{{\left| {{\bf{g}}_{di}^T{\bf{F}}{{\bf{\eta }}_r}} \right|}^2}}}{{{{\left| {{\bf{g}}_{di}^T{\bf{F}}{{\bf{h}}_{ui}}} \right|}^2}}}} \right\} \!-\! \frac{{\kappa _r^2\xi}}{{NP_U\sigma _{{h_{ui}}}^2}} \xrightarrow[{N \to \infty }]{}  0.
\end{align}
\begin{align}\label{C1}
&{\mathbb{E}}\left\{ {\frac{{{{\left| {{\bf{g}}_{di}^T{{\bf{\eta }}_t}} \right|}^2}}}{{{{\left| {{\bf{g}}_{di}^T{\bf{F}}{{\bf{h}}_{ui}}} \right|}^2}}}} \right\} - \frac{{\kappa _t^2{P_R}}}{{{N^4}\sigma _{{g_{di}}}^2\sigma _{{h_{ui}}}^4}} \xrightarrow[{N \to \infty }]{}  0,\notag\\
&{\mathbb{E}}\left\{ {\frac{{{{\left\| {{\bf{g}}_{di}^T{\bf{F}}{{\bf{G}}_{RR}}} \right\|}^2}}}{{{{\left| {{\bf{g}}_{di}^T{\bf{F}}{{\bf{h}}_{ui}}} \right|}^2}}}} \right\} - \frac{{\sigma _{LIr}^2}}{{\sigma _{{h_{iu}}}^2}} \xrightarrow[{N \to \infty }]{}  0,\notag\\
&{\mathbb{E}}\left\{ {\frac{1}{{{{\left| {{\bf{g}}_{di}^T{\bf{F}}{{\bf{h}}_{ui}}} \right|}^2}}}} \right\} - \frac{1}{{{N^4}\sigma _{{g_{di}}}^4\sigma _{{h_{ui}}}^4}} \xrightarrow[{N \to \infty }]{}  0.
\end{align}
Moreover, we have the following approximations
\begin{align}\label{rho ^2}
&{\mathbb{E}}\left\{ {{{\left\| {{\bf{FA}}} \right\|}^2}} \right\} - {N^3}\sum\limits_{i = 1}^K {\left( {\sigma _{{g_{ui}}}^4\sigma _{{h_{di}}}^2 + \sigma _{{g_{di}}}^2\sigma _{{h_{ui}}}^4} \right)}  \xrightarrow[{N \to \infty }]{}  0,\notag \\
&{\mathbb{E}}\left\{ {{{\left\| {{\bf{F}}{{\bf{G}}_{RR}}} \right\|}^2}} \right\} - {N^3}\sigma _{LIr}^2\sum\limits_{i = 1}^K {\left( {\sigma _{{g_{ui}}}^2\sigma _{{h_{di}}}^2 + \sigma _{{g_{di}}}^2\sigma _{{h_{ui}}}^2} \right)}   \xrightarrow[{N \to \infty }]{}  0,\notag \\
&{\mathbb{E}}\left\{ {{{\left\| {{\bf{F}}{{\bf{\eta }}_r}} \right\|}^2}} \right\} \!-\! \kappa _r^2{N^2}\xi \sum\limits_{i = 1}^K {\left( {\sigma _{{g_{ui}}}^2\sigma _{{h_{di}}}^2 \!+\! \sigma _{{g_{di}}}^2\sigma _{{h_{ui}}}^2} \right)}/{P_U} \xrightarrow[{N \to \infty }]{}  0,\notag \\
&{\mathbb{E}}\left\{ {{{\left\| {\bf{F}} \right\|}^2}} \right\} - \sigma _R^2{N^2}\sum\limits_{i = 1}^K {\left( {\sigma _{{g_{ui}}}^2\sigma _{{h_{di}}}^2 + \sigma _{{g_{di}}}^2\sigma _{{h_{ui}}}^2} \right)}  \xrightarrow[{N \to \infty }]{}  0,\notag \\
&{\mathbb{E}}\left\{ {{\rho ^2}} \right\} - \frac{{{P_R}}}{{{N^3}J}}\xrightarrow[{N \to \infty }]{}  0.
\end{align}
By substituting \eqref{C1}-\eqref{rho ^2} into \eqref{expression of EAi 1}, we can complete the proof after some simplifications.

%

\bibliographystyle{IEEEtran}
\bibliography{IEEEabrv,Ref}

\end{document}